\documentclass[11pt]{article}
\usepackage[margin=1in]{geometry}
\usepackage{setspace}
\usepackage{indentfirst}
\usepackage[T1]{fontenc}
\usepackage[utf8]{inputenc}
\usepackage{lmodern}
\usepackage{microtype}
\usepackage[american]{babel}
\usepackage{amsmath,amssymb,amsfonts,amsthm,mathtools}

\numberwithin{equation}{section}
\usepackage{graphicx}
\usepackage{tikz}
\usepackage{booktabs}
\usepackage{threeparttable}
\usepackage{tabularx}
\usepackage{array}
\usepackage{ragged2e}
\usepackage{placeins}
\numberwithin{figure}{section}
\usepackage[
font=small, labelfont=bf, labelsep=period ]{caption}
\usepackage{enumitem}
\setlist{
topsep=6pt, itemsep=3pt, parsep=0pt }
\theoremstyle{plain}

\newtheorem{proposition}{Proposition}
\newtheorem{theorem}{Theorem}
\newtheorem{corollary}{Corollary}
\newtheorem{assumption}{Assumption}
\theoremstyle{definition}
\newtheorem{definition}{Definition}
\theoremstyle{remark}
\newtheorem{remark}{Remark}
\usepackage[round,authoryear]{natbib}
\usepackage[
colorlinks=true, allcolors=blue ]{hyperref}
\usepackage{needspace}
\usepackage{chngcntr}
\counterwithin{table}{section}

\begin{document}
\title{\textbf{Correlation Ambiguity, Diversification Fragility, and Insurance Pricing}}
\author{Peiyun Deng\thanks{\ School of Economics and Finance, Shanghai International Studies University, Shanghai 201620, China. Email: \href{mailto:2022020@shisu.edu.cn}{2022020@shisu.edu.cn}.} \and Shunzhi Pang\thanks{\ School of Economics and Management, Tsinghua University, Beijing 100084, China; Faculty of Economics and Business, KU Leuven, Leuven 3000, Belgium. Email: \href{mailto:psz22@mails.tsinghua.edu.cn}{psz22@mails.tsinghua.edu.cn}.}}
\date{\today} 
\maketitle
\begin{abstract}
\noindent
We study how correlation ambiguity affects insurance supply, investment, and pricing in a competitive market. An insurer jointly chooses underwriting and risky investment under worst-case correlation uncertainty, with the insurance loading determined by insurance supply and household demand. With the correlation midpoint and all other model parameters fixed, widening the correlation interval makes diversification less reliable. With underwriting fixed, absolute financial investment weakly declines; with endogenous underwriting, investment may instead rise as insurance supply contracts. Nevertheless, joint exposure to uncertain dependence, measured by underwriting times absolute financial investment, strictly decreases whenever positive. Once this exposure vanishes, further ambiguity leaves the equilibrium unchanged. Insurance prices redistribute ambiguity costs between insurer and policyholder, while combined local certainty-equivalent surplus declines as long as joint exposure remains positive. 

\vspace{0.8em} 
\end{abstract}
\clearpage 

\section{Introduction}\label{sec:introduction}

Insurers assume underwriting risk while investing premiums and capital in financial markets. Their combined risk depends in part on the correlation between underwriting and financial shocks. Financial investment can offset adverse underwriting outcomes when the shocks move in opposite directions, but it can also amplify the risk of the combined position when they move together. Consequently, investment opportunities may affect not only insurers' financial portfolios but also their willingness to supply insurance. A growing literature shows that insurers' balance sheets and investment opportunities matter for insurance prices \citep{koijen2015cost, koijen2022fragility, knox2024insurers}.

Yet the dependence between underwriting and financial risks may itself be uncertain. Even when expected claims, expected financial returns, and their marginal volatilities are well specified, the correlation between these risks can be difficult to estimate from available data. The distinction between risk and ambiguity dates back to \citet{knight1921risk} and \citet{ellsberg1961risk}, while more recent research examines ambiguity about dependence specifically \citep{epstein2019ambiguous, bauch2025dependence}. This paper investigates how correlation ambiguity affects insurers' underwriting and investment decisions and, through market clearing, the quantity and price of insurance. 

To illustrate the difficulty of estimating underwriting--financial dependence, we use U.S. property--casualty statutory data from SNL and CRSP value-weighted total market returns. We construct a common claims-performance factor from unexpected insurer losses relative to earned premiums, with higher values indicating more favorable underwriting outcomes. The CRSP series provides a market-return benchmark rather than insurers' actual investment returns. Figure~\ref{fig:dependence} presents rolling correlations between the two series and their confidence intervals. These intervals are relatively wide and frequently include zero, suggesting considerable uncertainty about the estimated dependence. We use their median half-width as a reference for the magnitude of uncertainty in our numerical analysis, without interpreting the statistical intervals as insurers' subjective ambiguity sets.\footnote{Favorable claims performance is measured as minus incurred losses and loss-adjustment expenses divided by positive net premiums earned. Expanding insurer-effects forecasts use only prior quarters; forecast errors are scaled by their preceding 20-quarter standard deviation and aggregated using positive lagged admitted-asset weights. Monthly CRSP total returns are compounded to calendar quarters. Rolling 32-quarter correlations are evaluated using paired circular four-quarter block-bootstrap samples, with 2,000 draws and fifth/ninety-fifth percentile endpoints. First-stage estimates are held fixed during resampling. The empirical half-width is the median across rolling windows. }

\begin{figure}[!htbp]
\centering
\includegraphics[width=\linewidth]{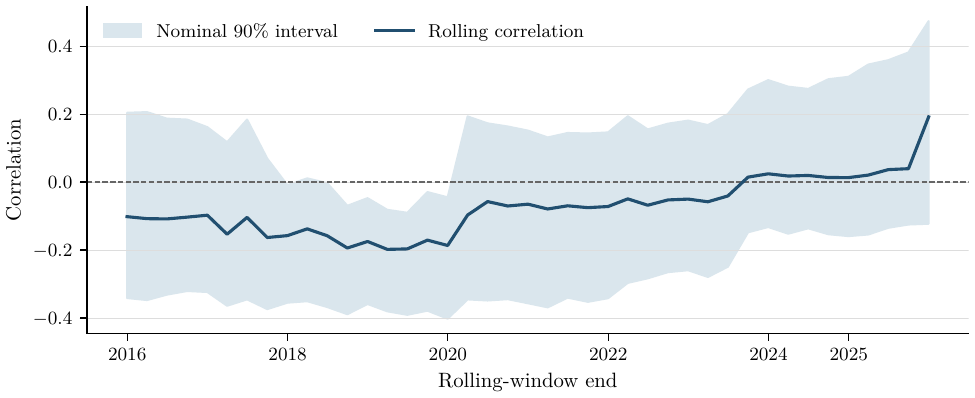}
\caption{Rolling underwriting--market dependence. The line shows the correlation between the common claims-performance factor and CRSP value-weighted total market returns over rolling 32-quarter windows. Shading represents nominal 90\% paired block-bootstrap confidence intervals.}
\label{fig:dependence}
\end{figure}

\FloatBarrier

To examine the market implications of uncertain dependence, we consider a competitive insurance market in which an insurer jointly chooses underwriting and financial investment, while a representative household chooses insurance coverage. The insurer evaluates terminal wealth under CARA utility and an interval of possible correlations between underwriting surplus and financial returns, while the household has local mean--variance preferences. Expected claims, financial investment opportunities, and marginal volatilities are held fixed, with the insurance loading determined by market clearing. We characterize a unique competitive equilibrium within the maintained loading domain, encompassing long and short financial positions, pure underwriting, and both underwriting boundaries. 

When underwriting is fixed, widening the correlation interval around an unchanged midpoint weakly reduces absolute financial investment. This need not hold in equilibrium: the insurer may contract underwriting sufficiently for financial investment to increase. Nevertheless, joint exposure, measured by the product of underwriting and absolute financial investment, strictly decreases under fixed-midpoint widening whenever initially positive, including across equilibrium regimes. Once joint exposure reaches zero, further widening leaves the allocation and insurance loading unchanged. The insurer can therefore reduce its exposure to uncertain dependence through underwriting adjustments.  

We also show that the competitive allocation maximizes the sum of the insurer's and household's local monetary certainty equivalents, with insurance loading payments canceling as transfers. This representation provides a basis for analyzing the direct cost of correlation ambiguity and its distribution through equilibrium insurance prices. A higher loading may offset part of the insurer's direct cost, but the corresponding transfer is borne by the household. Moreover, evaluating the equilibrium allocations at a common reference correlation within the smallest compared interval reveals a weak loss from allocation adjustment, distinct from the additional cost of worst-case evaluation.

Experimental studies document how insurers and actuaries respond to imprecise or conflicting information when pricing insurance and warranties \citep{hogarth1989risk, hogarth1992pricing, kunreuther1993insurer, kunreuther1995ambiguity, cabantous2007ambiguity, cabantous2011imprecise}. Economic premium principles link insurance prices to risk sharing \citep{buhlmann1980economic, buhlmann1984general}, while theoretical work examines ambiguity in insurance losses, insurability, premium principles, and competitive insurance pricing \citep{brockett1999underwriting, zhao2011ambiguity, fujii2017economic, anwar2012competitive, pichler2014insurance}. Related models incorporate insurance demand, premium competition, capital requirements, and disagreement across catastrophe-loss models \citep{emms2012equilibrium, li2021dynamic, dietz2019ambiguity, dietz2021pricing, li2026mitigating}. Rather than introducing uncertainty about the marginal distribution of insurance losses, we examine ambiguity about the dependence between underwriting and financial risks and its implications for competitive insurance pricing.

Insurer investment and reinsurance decisions have been examined under model uncertainty involving stochastic volatility, jumps, and asset mispricing \citep{yi2013robust, zeng2016robust, gu2017optimal}. Research on correlation ambiguity considers portfolio allocation, consumption, limited participation, and financial-market equilibrium \citep{fouque2016portfolio, han2020optimal, huang2017limited, sun2026limited}. More closely, \citet{cheng2024robust} establish investment contraction and zero holdings for an insurer with a given nontradable surplus, while \citet{jiang2022portfolio} study portfolio concentration, correlation-invariant allocations, and portfolio inertia. Related actuarial models examine ambiguous dependence in insurance, reinsurance, and investment decisions \citep{pun2016non, fadina2025optimal, xu2025optimal, wu2025alpha}. Our setting allows underwriting exposure to adjust jointly with financial investment and insurance prices, so the response of financial holdings need not coincide with that of joint risk exposure.

Insurance-price dynamics have been linked to aggregate capacity and financing frictions \citep{henriet2016dynamics}, recapitalization costs and insurers' risk appetite \citep{luciano2022fluctuations}, and heterogeneous beliefs about expected losses \citep{feng2023insurance}. Related work also considers robust insurance pricing with liquidity management and outside investment \citep{pang2026robust, chen2026robust}. Most closely, \citet{chen2025dynamic} jointly determine underwriting, investment, and competitive insurance prices under known correlation. \citet{kang2026stochastic} study correlation ambiguity in an insurer--reinsurer Stackelberg game, where reinsurance premiums and proportional reinsurance are chosen sequentially. In contrast, our model concerns a competitive primary-insurance market in which insurer underwriting supply clears against household coverage demand. 

The remainder of the paper is organized as follows. Section~\ref{sec:model} presents the model, and Section~\ref{sec:decisions} analyzes the insurer's optimal underwriting and investment decisions under correlation ambiguity. Section~\ref{sec:equilibrium} characterizes the competitive equilibrium and examines the implications for joint risk exposure and ambiguity-cost incidence. Section~\ref{sec:quantitative} provides numerical illustrations, and Section~\ref{sec:conclusion} concludes. Part of the proofs are collected in Appendix~\ref{app:proofs}.

\section{General Model and Problem Formulation} \label{sec:model}

Fix a finite horizon $T>0$ and a complete filtered probability space $(\Omega,\mathcal F,\{\mathcal F_t\}_{0\le t\le T},\mathbb P)$ satisfying the usual conditions. Let $B^I=\{B_t^I\}_{0\le t\le T}$ and $B^S=\{B_t^S\}_{0\le t\le T}$ be two independent standard Brownian motions with respect to this filtration. All stochastic processes are defined on this space, and $\mathbb E$ denotes expectation under $\mathbb P$.

\subsection{Insurance Market}

Consider a unit risk pool whose aggregate claims are approximated by a diffusion process $L=\{L_t\}_{0\le t\le T}$, with $L_0=0$, satisfying 
\begin{equation}
    \mathrm dL_t=l\,\mathrm dt-\eta\,\mathrm dB_t^I,\qquad l,\eta>0. \label{eq:loss}
\end{equation}
Here, $l$ is the expected claims rate and $\eta$ is the claims volatility per unit risk pool. As a diffusion approximation, $L$ need not have nondecreasing sample paths.

Insurance contracts renew continuously. Let $\theta=\{\theta_t\}_{0\le t\le T}$ denote the safety loading, and let $x=\{x_t\}_{0\le t\le T}$ be the insurer's underwriting policy, where $x_t\in[0,1]$ represents the proportion of the risk pool covered. The premium rate per unit coverage is $(1+\theta_t)l$, giving the insurer an underwriting cash flow of
\begin{equation*}
    x_t(1+\theta_t)l\,\mathrm dt-x_t\,\mathrm dL_t =\theta_t l x_t\,\mathrm dt+\eta x_t\,\mathrm dB_t^I.
\end{equation*} 
We refer to $\eta\,\mathrm dB_t^I$ as the underwriting surplus shock per unit coverage. A positive realization corresponds to claims below their expected rate and hence a favorable underwriting outcome. 

The representative household chooses a progressively measurable insurance coverage policy $\beta=\{\beta_t\}_{0\le t\le T}$, with $\beta_t\in[0,1]$. Financial investment is outside the household's decision problem. Let $N=\{N_t\}_{0\le t\le T}$ denote household wealth, with deterministic initial value $N_0=n\in\mathbb R$, and let $r\in\mathbb R$ be the risk-free interest rate. After paying insurance premiums and bearing retained claims, household wealth evolves, apart from other allocation-independent cash flows, according to 
\begin{equation*}
    \mathrm dN_t=(rN_t-l-\beta_t\theta_t l)\,\mathrm dt+(1-\beta_t)\eta\,\mathrm dB_t^I.
\end{equation*}

The household evaluates insurance coverage under local mean--variance preferences with risk-aversion parameter $\alpha>0$. At each date, it evaluates the conditional drift rate of wealth less $\frac{\alpha}{2}$ times its instantaneous variance rate. Omitting the allocation-independent term $rN_t-l$, the household's local monetary certainty-equivalent flow is
\begin{equation}
    \mathrm{CE}^{H}_t(\beta,\theta) =-\beta\theta l-\frac{\alpha\eta^2}{2}(1-\beta)^2. \label{eq:household-ce}
\end{equation} 
Related mean--variance specifications of insurance demand appear in \citet{luciano2022fluctuations}. Since the objective is strictly concave in $\beta$, the unique optimal coverage on the maintained loading domain is 
\begin{equation}
    \beta_t^*(\theta)=1-\frac{l\theta}{\alpha\eta^2}, \qquad 0\le\theta\le\frac{\alpha\eta^2}{l}. \label{eq:demand}
\end{equation}
The upper endpoint is the household's choke loading, at which insurance demand vanishes. We maintain this interval as the admissible loading domain. Following the competitive insurance-pricing setting of \citet{chen2025dynamic}, the loading is determined by market clearing between household demand and the insurer's optimal underwriting supply. 

\subsection{Financial Market}

The financial market consists of a risk-free asset and a risky asset. The risk-free asset price process $S^0=\{S_t^0\}_{0\le t\le T}$ earns a constant interest rate $r$ and satisfies
\begin{equation*}
    \mathrm dS_t^0=rS_t^0\,\mathrm dt,\qquad S_0^0=1.
\end{equation*} 
The risky asset price process $S=\{S_t\}_{0\le t\le T}$, with deterministic initial value $S_0>0$, follows
\begin{equation*}
    \frac{\mathrm dS_t}{S_t} =\mu\,\mathrm dt+\sigma\,\mathrm dW_t^{S,\rho},
\end{equation*} 
where $\mu>r$ is the expected return, $\sigma>0$ is the volatility, and $W^{S,\rho}$ denotes the Brownian return driver, constructed below to incorporate its dependence on underwriting risk.

Let $\rho=\{\rho_t\}_{0\le t\le T}$ be a progressively measurable correlation process taking values in $[\underline{\rho},\overline{\rho}]$, where $-1<\underline{\rho}\le\overline{\rho}<1$. Using the independent Brownian motions $B^I$ and $B^S$ introduced above, define $W^{S,\rho}=\{W_t^{S,\rho}\}_{0\le t\le T}$ by 
\begin{equation*}
    W_t^{S,\rho} =\int_0^t\rho_s\,\mathrm dB_s^I +\int_0^t\sqrt{1-\rho_s^2}\,\mathrm dB_s^S.
\end{equation*} 
The bounded progressive integrands define a continuous local martingale satisfying
\begin{equation*}
    \mathrm d\langle W^{S,\rho}\rangle_t=\mathrm dt, \qquad \mathrm d\langle B^I,W^{S,\rho}\rangle_t=\rho_t\,\mathrm dt.
\end{equation*} 
By L\'evy's characterization, $W^{S,\rho}$ is a standard Brownian motion with respect to the common filtration. This construction of correlated shocks from independent Brownian drivers follows the approach used in stochastic-correlation models such as \citet{teng2016modelling}. In our setting, the correlation process is selected from an admissible uncertainty set rather than assigned a prescribed stochastic law.  

The process $\rho_t$ represents the instantaneous correlation between the underwriting surplus shock and the risky-return shock. Since the stochastic component of claims in \eqref{eq:loss} has the opposite sign, its instantaneous correlation with financial returns is $-\rho_t$. A negative underwriting-surplus--return correlation therefore allows a long financial position to hedge underwriting risk. The dependence of the risky asset price process on $\rho$ is suppressed in the notation.

Let $M=\{M_t\}_{0\le t\le T}$ denote the insurer's total wealth, with initial value $M_0=m\in\mathbb R$. The insurer chooses a monetary investment policy $y=\{y_t\}_{0\le t\le T}$, where $y_t\in\mathbb R$ is the amount invested in the risky asset and $M_t-y_t$ is held in the risk-free account. Borrowing and short selling are permitted. Combining self-financing financial trading with underwriting cash flow gives 
\begin{equation}
    \begin{aligned}
        \mathrm dM_t
        &=x_t\big[(1+\theta_t)l\,\mathrm dt-\mathrm dL_t\big]
        +(M_t-y_t)r\,\mathrm dt
        +y_t\frac{\mathrm dS_t}{S_t}\\
        &=\big[rM_t+\theta_t l x_t+(\mu-r)y_t\big]\,\mathrm dt
        +\eta x_t\,\mathrm dB_t^I
        +\sigma y_t\,\mathrm dW_t^{S,\rho}.
        \end{aligned}
    \label{eq:wealth-dynamics}
\end{equation} 
The insurer's wealth is not constrained to be nonnegative. The model imposes no default, liquidity, regulatory capital, or financing-cost constraints.

\subsection{Correlation Ambiguity}

The insurer faces uncertainty about the dependence between underwriting and financial risks. For any deterministic restart date $t\in[0,T]$, let $\mathcal N_t$ denote the set of correlation processes $\rho=\{\rho_s\}_{t\le s\le T}$ that are progressively measurable with respect to the common filtration and satisfy $\rho_s\in[\underline{\rho},\overline{\rho}]$. We use interval-valued representatives in the diffusion coefficients. Correlation is the only uncertain coefficient, so the marginal drifts and volatilities of underwriting and financial returns remain unchanged. Stochastic integrals in a problem restarted at $t$ are taken over $[t,s]$.  

The set $\mathcal N_t$ is rectangular: it is stable under pasting at stopping times and on events in the corresponding filtration. Thus, Nature may adjust the adverse correlation as information arrives, rather than choosing a single correlation parameter for the entire horizon. The empirical confidence intervals discussed in Section \ref{sec:introduction} provide a reference for the magnitude of uncertainty but do not determine this dynamic specification.

We adopt a strong formulation in which Nature's progressively measurable correlation choices enter the diffusion coefficients on a fixed probability space. Related approaches to model uncertainty appear in \citet{sirbu2014strong} and, in the context of portfolio optimization with covariance uncertainty, \citet{guo2022robust}. Correlation ambiguity has also been studied in robust portfolio choice \citep{fouque2016portfolio} and insurer investment with a given nontradable surplus \citep{cheng2024robust}.\footnote{Controlled Brownian integrals also arise in path-law representations of $G$-expectations; see \citet{denis2011function}.} 

The insurer's admissible policies form the set $\mathcal A_t$ of progressively measurable pairs $(x,y)$ on $[t,T]$, with $x_s\in[0,1]$ and $y_s\in\mathbb R$, $\mathrm ds\times\mathbb P$-almost everywhere, satisfying 
\begin{equation}
    \mathbb E\exp\left(b\int_t^T y_s^2\,\mathrm ds\right)<\infty, \qquad \text{for every } b>0. \label{eq:admissibility}
\end{equation} 
This condition includes all bounded progressive investment policies and all deterministic square-integrable investments, without imposing a common bound on risky holdings. It provides the exponential integrability needed for the terminal CARA criterion and the martingale arguments used to verify optimality in Appendix~\ref{app:proof:verification}.

\subsection{Insurer Optimization and Market Equilibrium}

The insurer evaluates terminal wealth using constant absolute risk aversion (CARA) utility, $U(m)=-\frac{1}{\gamma}e^{-\gamma m}$, with risk-aversion parameter $\gamma>0$. We maintain the following assumptions on the model parameters and feasible choices.  

\begin{assumption} \label{ass:baseline}
The horizon $T>0$ is finite. The parameters $r,\mu,l,\eta,\sigma,\alpha,\gamma$ are constant and deterministic, with $r\in\mathbb R$, $\mu>r$, and $l,\eta,\sigma,\alpha,\gamma>0$. The correlation bounds satisfy $-1<\underline{\rho}\le\overline{\rho}<1$. The insurer's underwriting share and monetary investment satisfy $x_t\in[0,1]$ and $y_t\in\mathbb R$, respectively, while the safety loading is restricted to $0\le\theta_t\le\frac{\alpha\eta^2}{l}$.  
\end{assumption}

Fix a bounded continuous deterministic loading path $\theta$ in the maintained price domain. For any deterministic date $t\in[0,T]$ and initial wealth $m\in\mathbb R$, let $M^{t,m;x,y,\rho}=\{M_s^{t,m;x,y,\rho}\}_{t\le s\le T}$ denote the unique strong solution of \eqref{eq:wealth-dynamics} starting from $M_t=m$, under an admissible insurer policy $(x,y)\in\mathcal A_t$ and correlation process $\rho\in\mathcal N_t$. The insurer chooses underwriting and investment to maximize its worst-case expected terminal utility: 
\begin{equation}
    \mathcal V_t(m) = \operatorname*{ess\,sup}_{(x,y)\in\mathcal A_t} \operatorname*{ess\,inf}_{\rho\in\mathcal N_t} \mathbb E\left[ U\left(M_T^{t,m;x,y,\rho}\right) \,\middle|\,\mathcal F_t \right]. \label{eq:robust-problem}
\end{equation}

Both the insurer and the household take the loading as given when making their respective decisions. In equilibrium, the loading adjusts to equate the insurer's underwriting supply with household insurance demand.

\begin{definition}[\textbf{Competitive Insurance Equilibrium}] \label{def:equilibrium}
Under Assumption~\ref{ass:baseline}, a competitive insurance equilibrium is a tuple $(\theta^*,x^*,y^*,\beta^*)$ of continuous deterministic paths on $[0,T]$, with $\theta_t^*\in\left[0,\frac{\alpha\eta^2}{l}\right]$, $x_t^*,\beta_t^*\in[0,1]$, and $y_t^*\in\mathbb R$, such that:
\begin{enumerate} 
    \item For each date $t$, household coverage $\beta_t^*$ maximizes \eqref{eq:household-ce} at the loading $\theta_t^*$.
    \item For every deterministic restart date $t\in[0,T]$ and initial wealth $m\in\mathbb R$, the restriction of $(x^*,y^*)$ to $[t,T]$ belongs to $\mathcal A_t$ and attains \eqref{eq:robust-problem} under $\theta^*$ against the full class of admissible progressive policies. 
    \item Insurance supply equals demand at every date: $x_t^*=\beta_t^*$.
\end{enumerate}
\end{definition}

Theorem~\ref{thm:verification} establishes that the insurer's dynamic problem can be solved through pointwise optimization. Combining the resulting underwriting supply with household demand then determines the competitive equilibrium.

\section{Optimal Underwriting and Investment}\label{sec:decisions}

This section examines how correlation ambiguity affects the insurer's underwriting and investment decisions at a given insurance loading. We derive the worst-case local objective, characterize optimal financial investment conditional on underwriting, and examine its response to changes in the correlation interval. We then verify that the pointwise optimization solves the insurer's dynamic terminal-utility problem. 

\subsection{Worst-Case Correlation}

At a given date $t$, for a loading $\theta$ and correlation $\rho$, consider the insurer's local monetary certainty-equivalent objective 
\begin{equation}
    F_t^\rho(x,y;\theta)=\theta lx+(\mu-r)y-\frac{\gamma e^{r(T-t)}}{2}\left(\eta^2x^2+\sigma^2y^2+2\rho\eta\sigma xy\right). \label{eq:fixed-correlation}
\end{equation} 
The factor $\gamma e^{r(T-t)}$ reflects the risk-free accumulation of monetary exposures. Theorem~\ref{thm:verification} establishes the connection between this local criterion and the insurer's dynamic CARA problem. 

\begin{proposition} \label{prop:wedge}
Under Assumption~\ref{ass:baseline}, minimizing the objective over the correlation interval gives 
\begin{equation}
    \begin{aligned}
        F_t(x,y;\theta)
        & :=\min_{\rho\in[\underline{\rho},\overline{\rho}]}F_t^\rho(x,y;\theta)\\
        & =\theta lx+(\mu-r)y-\frac{\gamma e^{r(T-t)}}{2}\Big[ \eta^2x^2+\sigma^2y^2 +(\underline{\rho}+\overline{\rho})\eta\sigma xy +(\overline{\rho}-\underline{\rho})\eta\sigma x|y| \Big].
    \end{aligned}
\label{eq:wedge}
\end{equation}
The worst-case correlation is $\overline{\rho}$ when $xy>0$ and $\underline{\rho}$ when $xy<0$. When $xy=0$, the objective is independent of correlation. Relative to evaluating the same portfolio at the interval midpoint, the additional local certainty-equivalent cost is
\begin{equation}
    \frac{\gamma e^{r(T-t)}}{2}(\overline{\rho}-\underline{\rho})\eta\sigma x|y|. \label{eq:ambiguity-cost}
\end{equation} 
\end{proposition}

\begin{proof}\label{proof:prop:wedge}
Only the covariance term in \eqref{eq:fixed-correlation} depends on $\rho$. Since $x\ge0$, 
\begin{equation*}
    \max_{\rho\in[\underline{\rho},\overline{\rho}]}\rho xy =\frac{\underline{\rho}+\overline{\rho}}{2}xy +\frac{\overline{\rho}-\underline{\rho}}{2}x|y|.
\end{equation*} 
Substitution into \eqref{eq:fixed-correlation} gives \eqref{eq:wedge} and the worst-case correlation choices. Subtracting the resulting objective from its midpoint evaluation yields \eqref{eq:ambiguity-cost}. 
\end{proof}

\begin{remark}
The penalty in \eqref{eq:ambiguity-cost} is a local monetary flow cost evaluated at a fixed underwriting--investment allocation. It does not by itself measure the change in optimized certainty equivalent when the ambiguity interval expands, since underwriting and investment may adjust.
\end{remark}

\subsection{Conditional Investment and Exposure} \label{sec:conditional-exposure}

We first consider optimal financial investment for a given underwriting share. The insurer balances the expected excess return against portfolio risk, with the worst-case correlation depending on whether the financial position is long or short. This gives rise to a zero-investment region in which neither position is optimal.

\begin{proposition} \label{prop:soft}
Under Assumption~\ref{ass:baseline}, for a fixed underwriting share $x>0$, the unique optimal financial position is 
\begin{equation}
    y_t(x)=
    \begin{cases}
        \displaystyle \frac{\mu-r}{\gamma\sigma^2}e^{-r(T-t)}-\frac{\overline{\rho}\eta x}{\sigma}, & \displaystyle \frac{(\mu-r)e^{-r(T-t)}}{\gamma\eta\sigma x}>\overline{\rho}, \\[2mm]
        0, & \displaystyle \underline{\rho}\le\frac{(\mu-r)e^{-r(T-t)}}{\gamma\eta\sigma x}\le\overline{\rho}, \\[2mm]
        \displaystyle \frac{\mu-r}{\gamma\sigma^2}e^{-r(T-t)}-\frac{\underline{\rho}\eta x}{\sigma}, & \displaystyle \frac{(\mu-r)e^{-r(T-t)}}{\gamma\eta\sigma x}<\underline{\rho}.
    \end{cases}
    \label{eq:soft-threshold}
\end{equation} 
When underwriting is zero, the insurer holds the stand-alone optimal financial position 
\begin{equation}
    y_t(0)=\frac{\mu-r}{\gamma\sigma^2}e^{-r(T-t)}>0. \label{eq:zero-underwriting-investment}
\end{equation}
The joint underwriting--investment optimum has zero financial investment if and only if
\begin{equation}
    x_t^*=\min\left\{1,\max\left\{0,\frac{\theta l}{\gamma\eta^2}e^{-r(T-t)}\right\}\right\}>0, \quad \underline{\rho}\le\frac{(\mu-r)e^{-r(T-t)}}{\gamma\eta\sigma x_t^*}\le\overline{\rho}.\label{eq:no-investment}
\end{equation} 
\end{proposition}

\begin{proof}\label{proof:prop:soft}
For fixed $x>0$, the objective $F_t(x,y;\theta)$ is strictly concave in $y$. Its derivative on the positive half-line is $(\mu-r)-\gamma e^{r(T-t)}(\sigma^2y+\overline{\rho}\eta\sigma x)$, while on the negative half-line $\overline{\rho}$ is replaced by $\underline{\rho}$. Thus, zero investment is optimal precisely when 
\begin{equation}
    (\mu-r)-\gamma e^{r(T-t)}\overline{\rho}\eta\sigma x \le0\le (\mu-r)-\gamma e^{r(T-t)}\underline{\rho}\eta\sigma x. \label{eq:financial-subgradient}
\end{equation} 
Solving the first-order conditions on the two half-lines, together with the condition at zero, gives \eqref{eq:soft-threshold}. At $x=0$, the covariance term vanishes, yielding \eqref{eq:zero-underwriting-investment}.

Suppose the joint optimum has $y=0$. Maximizing the objective over underwriting then gives the clipped share in \eqref{eq:no-investment}. This share must be positive, since $(0,0)$ cannot be optimal when $\mu-r>0$. The financial subgradient condition gives the remaining inequality in \eqref{eq:no-investment}.

Conversely, suppose \eqref{eq:no-investment} holds and choose the supporting correlation $\rho=\frac{(\mu-r)e^{-r(T-t)}}{\gamma\eta\sigma x_t^*}$, which belongs to the ambiguity interval. At $(x_t^*,0)$, the financial derivative of $F_t^\rho$ vanishes, and the underwriting coordinate satisfies the constrained first-order condition. Since $|\rho|<1$, $F_t^\rho$ is strictly concave and $(x_t^*,0)$ is its global maximizer. Moreover, $F_t\le F_t^\rho$ everywhere, with equality at this allocation. Hence $(x_t^*,0)$ also maximizes the robust objective.
\end{proof}

We next examine how the conditional financial position responds to a wider correlation interval. Write $\underline{\rho}=\rho_0-\phi$ and $\overline{\rho}=\rho_0+\phi$, where $\rho_0$ is the midpoint and $0\le\phi<1-|\rho_0|$. We hold underwriting, the midpoint, and all other primitives fixed.

\begin{corollary} \label{cor:conditional-ambiguity}
Under Assumption~\ref{ass:baseline}, fix the date and underwriting share $x>0$. On the strict positive-investment branch, where $\frac{(\mu-r)e^{-r(T-t)}}{\gamma\eta\sigma x}>\overline{\rho}$, the financial position satisfies $\frac{\partial y_t(x)}{\partial\overline{\rho}}=-\frac{\eta}{\sigma}x$. Along a fixed-midpoint widening, 
\begin{equation}
    \frac{\partial y_t(x)}{\partial\phi} =-\frac{\eta}{\sigma}x<0, \qquad \frac{\partial^2 y_t(x)}{\partial\phi\,\partial x} =-\frac{\eta}{\sigma}<0. \label{eq:conditional-width-derivatives}
\end{equation} 
More generally, the absolute financial position $|y_t(x)|$ is weakly decreasing in $\phi$ across all investment regimes, including the zero-investment region and the short-hedging branch.
\end{corollary}

\begin{proof}\label{proof:cor:conditional-ambiguity}
On the positive-investment branch of \eqref{eq:soft-threshold}, $\overline{\rho}$ increases with $\phi$, giving the derivatives in \eqref{eq:conditional-width-derivatives}. On the negative-investment branch, $\underline{\rho}$ decreases with $\phi$, so $\frac{\partial y_t(x)}{\partial\phi}=\frac{\eta}{\sigma}x>0$ moves the short position toward zero. When the conditional threshold lies inside the expanding correlation interval, investment is zero and remains so under further widening. Continuity at the regime boundaries establishes the global weak contraction of absolute exposure. At $x=0$, the financial position in \eqref{eq:zero-underwriting-investment} is independent of the correlation interval.
\end{proof}

Zero financial investment under correlation ambiguity also arises when the insurer has a given nontradable surplus \citep{cheng2024robust}, while related portfolio concentration has been studied in multiasset settings \citep{jiang2022portfolio}. Here, underwriting is itself a choice variable. The conditional contraction established above therefore provides a benchmark for the equilibrium analysis in Section~\ref{sec:equilibrium}, where underwriting and the market-clearing loading are jointly determined.

\subsection{CARA Value Function and Verification}

For a given insurance loading, define the insurer's optimized local monetary certainty-equivalent flow by 
\begin{equation*}
    \mathrm{CE}^{I}_t(\theta;\underline{\rho},\overline{\rho}) =\max_{\substack{0\le x\le1\\y\in\mathbb R}}F_t(x,y;\theta).
\end{equation*} 
We suppress the correlation endpoints when they remain fixed. The following theorem establishes that maximizing the local objective also solves the insurer's dynamic robust problem under terminal CARA utility. The verification uses the exponential integrability condition in \eqref{eq:admissibility}, allowing investment positions without a common exogenous bound.

\begin{theorem}[\textbf{Verification Theorem}] \label{thm:verification}
Under Assumption~\ref{ass:baseline}, let $\theta$ be a bounded continuous deterministic loading path. The local objective $F_t$ admits a unique maximizer $(x_t^*,y_t^*)$, which is deterministic and continuous in $t$. The resulting policy is admissible and attains the robust value in \eqref{eq:robust-problem}, given by 
\begin{equation}
    V(t,m)=-\frac{1}{\gamma}\exp\left\{-\gamma e^{r(T-t)}m-\int_t^T\gamma e^{r(T-s)}\mathrm{CE}^{I}_s(\theta_s)\,\mathrm ds\right\}. \label{eq:value}
\end{equation} 
The value function is $C^{1,2}$, and the optimal financial position satisfies
\begin{equation}
    |y_t^*|\le\frac{(\mu-r)e^{\max\{-r,0\}T}}{\gamma\sigma^2}+\frac{\eta}{\sigma}\max\{|\underline{\rho}|,|\overline{\rho}|\}. \label{eq:optimizer-bound} 
\end{equation} 
Among continuous deterministic policies, dynamic optimality is equivalent to pointwise maximization of $F_t$. 
\end{theorem}

\begin{proof}
See Appendix~\ref{app:proof:verification}.
\end{proof}

The factor $\gamma e^{r(T-t)}$ accounts for the accumulation of risk to the terminal horizon. With deterministic primitives and CARA utility, the optimal allocation does not depend on current wealth. In particular, the investment bound in \eqref{eq:optimizer-bound} ensures admissibility of the maximizing policy. When $r=0$, the time-to-maturity adjustment disappears, yielding time-independent equilibrium allocations and loadings in Section~\ref{sec:equilibrium}. 

\begin{remark} \label{remark:rectangular}
The rectangular uncertainty set permits correlation to vary over time, unlike a model with a single unknown constant correlation. To compare the two specifications, consider an interval $[\rho_0-\phi,\rho_0+\phi]\subset(-1,1)$ with midpoint $\rho_0$ and half-width $\phi\ge0$. Fix a bounded continuous deterministic loading path and an admissible deterministic policy $\{(x_s,y_s)\}_{t\le s\le T}$. Relative to evaluating the same policy at the midpoint, the ambiguity costs in terminal monetary certainty equivalent satisfy 
\begin{equation*}
    \gamma\phi\eta\sigma\left|\int_t^T e^{2r(T-s)}x_sy_s\,\mathrm ds\right|
    \le
    \gamma\phi\eta\sigma\int_t^T e^{2r(T-s)}x_s|y_s|\,\mathrm ds.
\end{equation*} 
The left-hand side is the ambiguity cost when correlation is one unknown constant over $[t,T]$, whereas the right-hand side is the cost under progressive selection from $\mathcal N_t$. For the deterministic policy considered here, the rectangular worst case is attained by choosing $\rho_s=\rho_0+\phi\operatorname{sign}(y_s)$, with any admissible correlation when $x_sy_s=0$. An exponential-martingale bound applies to all progressive correlation processes. By the triangle inequality, the right-hand side is no smaller than the left, with equality when $x_sy_s$ does not change sign and possible strict inequality when positive and negative exposures offset each other. 
\end{remark}

\section{Competitive Insurance Equilibrium} \label{sec:equilibrium}

This section characterizes the competitive insurance equilibrium and examines how it responds to correlation ambiguity. We first determine equilibrium insurance prices and allocations by combining the insurer's underwriting supply with household demand. We then study the adjustment of joint risk exposure through a local certainty-equivalent surplus representation and examine how ambiguity costs are distributed through insurance pricing.

\subsection{Equilibrium Pricing and Allocation}\label{sec:supply-clearing}

Following the risk-sharing approach to insurance pricing in \citet{buhlmann1980economic,buhlmann1984general}, we determine the insurance loading by equating the insurer's underwriting supply with household coverage demand. For a given loading $\theta$, let $x_t(\theta)$ denote the insurer's optimal underwriting share, with the associated financial position $y_t(x_t(\theta))$ characterized in Proposition~\ref{prop:soft}.

After optimizing over financial investment, the marginal certainty-equivalent value of underwriting is 
\begin{equation}
    \frac{\mathrm d}{\mathrm dx}F_t(x,y_t(x);\theta) =\theta l-\gamma e^{r(T-t)}\eta^2x -\frac{\gamma e^{r(T-t)}\eta\sigma}{2} \left[(\underline{\rho}+\overline{\rho})y_t(x)+(\overline{\rho}-\underline{\rho})|y_t(x)|\right]. \label{eq:supply-marginal}
\end{equation} 
This marginal return is continuous and strictly decreasing in $x$. Accordingly, optimal underwriting is zero when the marginal return at zero is nonpositive, equals the entire risk pool when the marginal return at one is nonnegative, and otherwise is determined by the unique interior solution. The resulting supply $x_t(\theta)$ is continuous in the loading and nondecreasing in $\theta$. 

Household demand in \eqref{eq:demand} decreases continuously from full coverage to zero over the maintained loading interval $[0,\frac{\alpha\eta^2}{l}]$. Since insurer supply is continuous, nondecreasing, and bounded between zero and one, supply and demand intersect at a unique loading $\theta_t^*$ satisfying 
\begin{equation}
    x_t(\theta_t^*)=1-\frac{l\theta_t^*}{\alpha\eta^2}, \qquad 0\le\theta_t^*\le\frac{\alpha\eta^2}{l}. \label{eq:clearing}
\end{equation}

The uniqueness of the equilibrium loading relies on this maintained price domain. If household demand were extended by clipping it to $[0,1]$, additional prices could support boundary allocations. In particular, no-trade outcomes may admit loadings above the household's choke loading, while strict full underwriting can be supported by negative loadings satisfying 
\begin{equation*}
    \frac{\gamma e^{r(T-t)}\eta^2(1-\overline{\rho}^{\,2})+\frac{\overline{\rho}\eta(\mu-r)}{\sigma}}{l} \le\theta\le0.
\end{equation*}  
The lower bound is negative in the strict full-underwriting regime. If nonnegative premiums are required, this range must also satisfy $\theta\ge-1$, with strict inequality for positive premiums. These possible boundary prices are excluded from the maintained loading domain. The singleton-correlation benchmark considered here retains the same underwriting and loading constraints as the ambiguity model. 

We can now characterize the equilibrium allocation and loading explicitly. Depending on the correlation interval, the insurer may cover the entire pool, provide partial coverage with a positive or negative financial position, hold no risky financial asset, or withdraw from underwriting altogether.

\begin{theorem}[\textbf{Complete Competitive Equilibrium}] \label{thm:regimes}
Under Assumption~\ref{ass:baseline}, there exists a unique competitive equilibrium within the maintained loading domain, with continuous deterministic loading and allocation paths. At each date $t$, exactly one of the following five cases applies, and household coverage satisfies $\beta_t^*=x_t^*$. 
\begin{enumerate}
\item[(i)] \textit{Full underwriting.} If 
\begin{equation}
\overline{\rho}\le \frac{(\mu-r)e^{-r(T-t)}}{2\gamma\eta\sigma} -\sqrt{1+\left(\frac{(\mu-r)e^{-r(T-t)}}{2\gamma\eta\sigma}\right)^2}, \label{eq:regime-full-condition}
\end{equation}
then
\begin{equation}
    \theta_t^*=0,\qquad x_t^*=1,\qquad y_t^*=\frac{\mu-r}{\gamma\sigma^2}e^{-r(T-t)}-\frac{\overline{\rho}\eta}{\sigma}>0. \label{eq:regime-full}
\end{equation} 
\item[(ii)] \textit{Interior underwriting with positive investment.} If 
\begin{equation}
    \begin{gathered}
        \frac{(\mu-r)e^{-r(T-t)}}{2\gamma\eta\sigma} -\sqrt{1+\left(\frac{(\mu-r)e^{-r(T-t)}}{2\gamma\eta\sigma}\right)^2} <\overline{\rho},\\
        \overline{\rho}< \min\left\{ \frac{(\mu-r)\left(\frac{1}{\alpha}+\frac{1}{\gamma}e^{-r(T-t)}\right)}{\eta\sigma}, \frac{\alpha\eta\sigma}{\mu-r} \right\}, 
    \end{gathered}
    \label{eq:regime-long-condition}
\end{equation}
then 
\begin{equation}
    \begin{aligned}
        x_t^* &=\frac{\frac{1}{\gamma}e^{-r(T-t)}\left(1-\frac{\overline{\rho}(\mu-r)}{\alpha\eta\sigma}\right)}{\frac{1-\overline{\rho}^{\,2}}{\alpha}+\frac{1}{\gamma}e^{-r(T-t)}},\\[1mm]
        \theta_t^* &=\frac{\alpha\eta^2}{l}(1-x_t^*),\\[1mm] 
        y_t^* &=\frac{(\mu-r)e^{-r(T-t)}}{\gamma\sigma^2} -\frac{\overline{\rho}\eta x_t^*}{\sigma}>0.
    \end{aligned}
    \label{eq:regime-long}
\end{equation} 
Both underwriting and loading are strictly interior.

\item[(iii)] \textit{Pure underwriting.} If 
\begin{equation}
    \underline{\rho}\le \frac{(\mu-r)\left(\frac{1}{\alpha}+\frac{1}{\gamma}e^{-r(T-t)}\right)}{\eta\sigma} \le\overline{\rho}, \label{eq:regime-pure-condition}
\end{equation}
then
\begin{equation}
    \begin{aligned}
        x_t^* &=\frac{\frac{1}{\gamma}e^{-r(T-t)}}
        {\frac{1}{\alpha}+\frac{1}{\gamma}e^{-r(T-t)}},\\[1mm]
        \theta_t^* &=\frac{\eta^2}{l\left(\frac{1}{\alpha}+\frac{1}{\gamma}e^{-r(T-t)}\right)},\\[1mm] 
        y_t^*&=0.
    \end{aligned}
    \label{eq:regime-pure}
\end{equation} 
\item[(iv)] \textit{Interior underwriting with negative investment.} If
\begin{equation}
    \frac{(\mu-r)\left(\frac{1}{\alpha}+\frac{1}{\gamma}e^{-r(T-t)}\right)}{\eta\sigma} <\underline{\rho}, \label{eq:regime-short-condition}
\end{equation}
then
\begin{equation}
    \begin{aligned}
        x_t^* &=\frac{\frac{1}{\gamma}e^{-r(T-t)}\left(1-\frac{\underline{\rho}(\mu-r)}{\alpha\eta\sigma}\right)}{\frac{1-\underline{\rho}^{\,2}}{\alpha}+\frac{1}{\gamma}e^{-r(T-t)}},\\[1mm]
        \theta_t^* &=\frac{\alpha\eta^2}{l}(1-x_t^*),\\[1mm]
        y_t^* &=\frac{(\mu-r)e^{-r(T-t)}}{\gamma\sigma^2} -\frac{\underline{\rho}\eta x_t^*}{\sigma}<0.
    \end{aligned}
    \label{eq:regime-short}
\end{equation}
Both underwriting and loading are strictly interior. 
\item[(v)] \textit{Zero underwriting.} If 
\begin{equation}
    \overline{\rho}\ge\frac{\alpha\eta\sigma}{\mu-r}, \label{eq:regime-zero-condition}
\end{equation} 
then
\begin{equation}
    \theta_t^*=\frac{\alpha\eta^2}{l},\qquad
    x_t^*=0,\qquad
    y_t^*=\frac{\mu-r}{\gamma\sigma^2}e^{-r(T-t)}>0.
\label{eq:regime-zero}
\end{equation} 
\end{enumerate} 
The five cases are mutually exclusive and exhaustive. The equilibrium allocation and the insurer's optimized local certainty equivalent are continuous across all reachable regime boundaries. Equality at the full-underwriting boundary is assigned to case (i), equality at either pure-underwriting boundary to case (iii), and equality at the zero-underwriting boundary to case (v).
\end{theorem}

\begin{proof}
See Appendix~\ref{app:proof:regimes}.
\end{proof}

The equilibrium financial position depends on which correlation endpoint is adverse. Positive investment is governed by the upper endpoint, whereas a short hedge is governed by the lower endpoint. In the pure-underwriting regime, the insurer holds no risky financial asset, and correlation no longer affects its local objective.

Unlike the conditional investment response in Section~\ref{sec:conditional-exposure}, the equilibrium financial position need not decrease when the correlation interval widens. The insurer can also adjust underwriting, which changes the amount of risk to be hedged or combined with financial investment.

\begin{remark} \label{remark:investment-rise}
Consider a fixed-midpoint widening of the correlation interval, holding the date and all other primitives fixed. Within the strict interior positive-investment regime of Theorem~\ref{thm:regimes}, differentiation of \eqref{eq:regime-long} shows that financial investment increases if and only if 
\begin{equation*}
    2\overline{\rho}\frac{\mu-r}{\eta\sigma} \left(\alpha+\gamma e^{r(T-t)}\right) > \alpha\left[\alpha+\gamma e^{r(T-t)}(1+\overline{\rho}^{\,2})\right].
\end{equation*} 
In particular, $\overline{\rho}>0$ is necessary. The financial first-order condition gives 
\begin{equation*}
    \frac{\mathrm dy_t^*}{\mathrm d\phi} =-\frac{\eta}{\sigma} \left(x_t^*+\overline{\rho}\frac{\mathrm dx_t^*}{\mathrm d\phi}\right).
\end{equation*} 
The first term reflects the direct effect of a wider correlation interval, while the second captures the adjustment of underwriting. Investment rises only when underwriting contracts sufficiently to outweigh the direct effect. Since $y_t^*>0$ in this regime, signed and absolute financial exposure increase together.

For a zero midpoint, a sufficient condition for an increasing-investment segment is 
\begin{equation*}
    \frac{\mu-r}{\sigma}>\alpha\eta.
\end{equation*} 
Under this condition, the zero-underwriting boundary $\phi=\frac{\alpha\eta\sigma}{\mu-r}$ lies below one, whereas \eqref{eq:ordering} places the pure-underwriting threshold above one. The equilibrium therefore remains in the interior positive-investment regime immediately before the zero-underwriting boundary, where the left derivative of investment with respect to $\phi$ is positive. Hence an increasing segment exists. This condition is sufficient, but not necessary for investment to increase more generally.
\end{remark}

\subsection{Local Surplus and Joint Exposure} \label{sec:global-exposure}

We next examine equilibrium risk allocation through the sum of the insurer's and household's local monetary certainty equivalents. At market clearing, the insurance loading acts as a transfer between the two parties, allowing the competitive allocation to be characterized by a common surplus objective. 

Substituting the household's optimal coverage \eqref{eq:demand} into \eqref{eq:household-ce} gives 
\begin{equation}
    \mathrm{CE}^{H,*}_t(\theta) =-l\theta+\frac{l^2\theta^2}{2\alpha\eta^2}, \qquad \frac{\mathrm d\mathrm{CE}^{H,*}_t}{\mathrm d\theta} =-l\left(1-\frac{l\theta}{\alpha\eta^2}\right). \label{eq:ce-levels}
\end{equation} 
At equilibrium, denote the insurer's, household's, and total local certainty-equivalent flows by 
\begin{equation*}
    \begin{aligned}
        \mathrm{CE}^{I,\mathrm{eq}}_t &=\mathrm{CE}^{I}_t(\theta_t^*;\underline{\rho},\overline{\rho}),\\
        \mathrm{CE}^{H,\mathrm{eq}}_t &=\mathrm{CE}^{H,*}_t(\theta_t^*),\\
        \mathrm{CE}^{T,\mathrm{eq}}_t &=\mathrm{CE}^{I,\mathrm{eq}}_t+\mathrm{CE}^{H,\mathrm{eq}}_t.
    \end{aligned}
\end{equation*}

At market clearing, the insurer's underwriting share equals household coverage. The loading payment therefore enters their local objectives with opposite signs and cancels when the two are added. Omitting allocation-independent terms, the resulting local surplus is 
\begin{equation}
    \begin{aligned}
        \mathrm{CE}^{T}_t(x,y;\underline{\rho},\overline{\rho}) ={}&(\mu-r)y-\frac{\alpha\eta^2}{2}(1-x)^2\\
        &-\frac{\gamma e^{r(T-t)}}{2} \Big[ \eta^2x^2+\sigma^2y^2 +(\underline{\rho}+\overline{\rho})\eta\sigma xy +(\overline{\rho}-\underline{\rho})\eta\sigma x|y| \Big].
    \end{aligned}
\label{eq:local-total-ce} 
\end{equation} 
This criterion combines the household's local mean--variance preferences with the insurer's certainty-equivalent representation of terminal CARA utility. It provides a local measure of risk-sharing surplus, rather than a certainty equivalent of pooled terminal wealth or a complete intertemporal welfare measure. This allocation-based perspective is related to economic premium principles \citep{buhlmann1984general} and work on premium principles under ambiguous losses \citep{fujii2017economic}. 

\begin{theorem}[\textbf{Local Surplus Maximization and Decentralization}] \label{thm:planner}
Under Assumption~\ref{ass:baseline}, the competitive allocation uniquely maximizes \eqref{eq:local-total-ce} over $[0,1]\times\mathbb R$. Conversely, the unique maximizing allocation can be supported as a competitive equilibrium by the loading 

\begin{equation}
    \theta_t^*=\frac{\alpha\eta^2}{l}(1-x_t^*). \label{eq:decentralizing-loading}
\end{equation} 
Consequently, equilibrium local certainty-equivalent surplus is weakly decreasing under nested expansions of the correlation interval.
\end{theorem}

\begin{proof}
See Appendix~\ref{app:proof:planner}.
\end{proof}

The surplus representation also allows us to compare equilibrium allocations as correlation ambiguity increases. Fix a date $t$, a midpoint $\rho_0\in(-1,1)$, and all non-correlation primitives. We consider a widening correlation interval of the form
\begin{equation}
    [\underline{\rho},\overline{\rho}] =[\rho_0-\phi,\rho_0+\phi], \qquad 0\le\phi<1-|\rho_0|. \label{eq:symmetric-experiment}
\end{equation} 
The half-width $\phi$ measures the extent of correlation ambiguity. For a given underwriting--investment allocation, the absolute sensitivity of instantaneous portfolio variance to correlation is 
\begin{equation*}
    \left|\frac{\partial}{\partial\rho} \left(\eta^2x^2+\sigma^2y^2+2\rho\eta\sigma xy\right) \right| =2\eta\sigma x|y|.
\end{equation*} 
We refer to this sensitivity as \emph{diversification fragility}. Up to the constant factor $2\eta\sigma$, it is measured by $x|y|$. This exposure arises both when the financial position hedges underwriting risk and when it adds positively correlated risk.

\begin{proposition} \label{prop:joint-exposure}
Under Assumption~\ref{ass:baseline}, fix the date, the correlation midpoint $\rho_0$, and all non-correlation primitives. Let $(x_t^*(\phi),y_t^*(\phi))$ denote the competitive allocation corresponding to \eqref{eq:symmetric-experiment}. For any $0\le\phi_1<\phi_2<1-|\rho_0|$, 
\begin{equation*}
    x_t^*(\phi_2)|y_t^*(\phi_2)| \le x_t^*(\phi_1)|y_t^*(\phi_1)|.
\end{equation*} 
The inequality is strict whenever the initial joint exposure is positive. If the initial joint exposure is zero, the equilibrium allocation and supporting loading remain unchanged under every further feasible widening. These comparisons hold across all equilibrium regimes, including investment and underwriting boundaries.
\end{proposition}

\begin{proof}\label{proof:prop:joint-exposure}
For $i=1,2$, let $(x_i,y_i)=(x_t^*(\phi_i),y_t^*(\phi_i))$. From \eqref{eq:local-total-ce}, the local surplus under a fixed-midpoint correlation interval satisfies 
\begin{equation*}
    \mathrm{CE}^{T}_t(x,y;\rho_0-\phi,\rho_0+\phi) =\mathrm{CE}^{T}_t(x,y;\rho_0,\rho_0) -\gamma e^{r(T-t)}\eta\sigma\phi x|y|.
\end{equation*} 
By Theorem~\ref{thm:planner}, each equilibrium allocation maximizes the corresponding local surplus. Hence
\begin{align*}
    \mathrm{CE}^{T}_t(x_1,y_1;\rho_0-\phi_1,\rho_0+\phi_1)  &\ge \mathrm{CE}^{T}_t(x_2,y_2;\rho_0-\phi_1,\rho_0+\phi_1),\\
    \mathrm{CE}^{T}_t(x_2,y_2;\rho_0-\phi_2,\rho_0+\phi_2) &\ge \mathrm{CE}^{T}_t(x_1,y_1;\rho_0-\phi_2,\rho_0+\phi_2).
\end{align*}
Adding these inequalities and cancelling the common terms gives 
\begin{equation*}
    \gamma e^{r(T-t)}\eta\sigma(\phi_2-\phi_1) \left(x_1|y_1|-x_2|y_2|\right)\ge0.
\end{equation*} 
Since $\phi_2>\phi_1$, joint exposure is weakly decreasing.

To establish strictness, suppose $x_1|y_1|=x_2|y_2|>0$. Both optimality inequalities must then hold with equality. Uniqueness implies that the two allocations coincide at some $(x,y)$ with $x>0$ and $y\ne0$. The financial first-order conditions at the two ambiguity widths therefore give 
\begin{equation*}
    \gamma e^{r(T-t)}\eta\sigma(\phi_2-\phi_1)x\operatorname{sign}(y)=0,
\end{equation*} 
which is impossible. Thus joint exposure strictly decreases whenever it is initially positive.

Finally, if $x_1|y_1|=0$, widening the interval leaves the local surplus at $(x_1,y_1)$ unchanged while weakly reducing it at every alternative allocation. Hence $(x_1,y_1)$ remains the unique maximizer for all larger feasible widths, and its supporting loading is unchanged by \eqref{eq:decentralizing-loading}.
\end{proof}

The result applies to finite changes in ambiguity, including transitions between investment and underwriting regimes. It does not require the equilibrium allocation to be differentiable at the regime boundaries. To distinguish allocation adjustment from changes in worst-case evaluation, we compare the equilibrium allocations at a common reference correlation.

\begin{corollary} \label{cor:common-correlation}
Under Assumption~\ref{ass:baseline}, fix the date, correlation midpoint $\rho_0$, and all non-correlation primitives. For any $0\le\phi_1\le\phi_2<1-|\rho_0|$ and fixed reference correlation $\rho^\dagger\in[\rho_0-\phi_1,\rho_0+\phi_1]$, the corresponding robust competitive allocations satisfy 
\begin{equation*}
    \mathrm{CE}^{T}_t\big(x_t^*(\phi_2),y_t^*(\phi_2);\rho^\dagger,\rho^\dagger\big) \le\mathrm{CE}^{T}_t\big(x_t^*(\phi_1),y_t^*(\phi_1);\rho^\dagger,\rho^\dagger\big). 
\end{equation*} 
Both allocations are evaluated at the same correlation $\rho^\dagger$. 
\end{corollary}

\begin{proof}
On a positive-investment branch, total local CE evaluated at $\rho^\dagger$ exceeds the robust equilibrium surplus by $\gamma e^{r(T-t)}\eta\sigma(\overline{\rho}-\rho^\dagger)x_t^*|y_t^*|$. On a negative-investment branch, the corresponding difference uses $\rho^\dagger-\underline{\rho}$ instead. Along the fixed-midpoint widening in \eqref{eq:symmetric-experiment}, the envelope derivative of robust equilibrium surplus is $-\gamma e^{r(T-t)}\eta\sigma x_t^*|y_t^*|$. Since both endpoint distances increase at unit rate, differentiation on each smooth investment branch gives 
\begin{equation}
\begin{aligned}
\frac{\mathrm d}{\mathrm d\phi}
\mathrm{CE}^{T}_t(x_t^*,y_t^*;\rho^\dagger,\rho^\dagger) =\gamma e^{r(T-t)}\eta\sigma
\frac{\mathrm d}{\mathrm d\phi}(x_t^*|y_t^*|) \times
\begin{cases}
\overline{\rho}-\rho^\dagger, & y_t^*>0,\\
\rho^\dagger-\underline{\rho}, & y_t^*<0.
\end{cases}
\end{aligned}
\label{eq:common-correlation-derivative}
\end{equation}

Because $\rho^\dagger$ belongs to the smaller interval, both endpoint distances are nonnegative. Proposition~\ref{prop:joint-exposure} implies that the derivative of joint exposure is nonpositive on every smooth branch, so the reference certainty equivalent is nonincreasing. Under pure or zero underwriting, the equilibrium allocation and its reference certainty equivalent remain constant. Finally, the continuity and finite regime structure established in Theorem~\ref{thm:regimes} extend the comparison across all regime boundaries, using the appropriate one-sided limits.
\end{proof}

\subsection{Ambiguity Costs and Price Incidence} \label{sec:loading-incidence}

At a fixed loading, widening the correlation interval weakly lowers the insurer's optimized local certainty equivalent. In equilibrium, the loading also adjusts, changing how the cost of ambiguity is shared between insurer and household.

\begin{proposition} \label{prop:fixedprice}
Under Assumption~\ref{ass:baseline}, enlarging the correlation interval at a fixed loading weakly reduces the insurer's optimized local certainty equivalent $\mathrm{CE}^{I}_t$. At the unique optimal allocation $(x,y)$, its envelope derivatives are 
\begin{align}
    \frac{\partial\mathrm{CE}^{I}_t}{\partial\theta}
    &=lx,\nonumber\\
    \frac{\partial\mathrm{CE}^{I}_t}{\partial\overline{\rho}}
    &=-\gamma e^{r(T-t)}\eta\sigma x\max(y,0),\nonumber\\
    \frac{\partial\mathrm{CE}^{I}_t}{\partial\underline{\rho}}
    &=-\gamma e^{r(T-t)}\eta\sigma x\min(y,0). 
    \label{eq:envelopes}
\end{align} 
At loading boundaries or singleton correlation intervals, these expressions apply along feasible parameter perturbations.
\end{proposition}

\begin{proof}
See Appendix~\ref{app:proof:fixedprice}.
\end{proof}

Along the fixed-midpoint widening in \eqref{eq:symmetric-experiment}, the envelope derivatives imply a direct marginal ambiguity cost of $\gamma e^{r(T-t)}\eta\sigma x|y|$ at the fixed-loading optimum. This cost depends on the insurer's joint exposure, regardless of whether its financial position is long or short. When the loading adjusts to clear the insurance market, its effect must also be taken into account.

\begin{proposition} \label{prop:incidence}
Under Assumption~\ref{ass:baseline}, consider the competitive equilibrium along the fixed-midpoint widening in \eqref{eq:symmetric-experiment}. The equilibrium local certainty-equivalent flows satisfy 
\begin{align}
    \frac{\mathrm d\mathrm{CE}^{I,\mathrm{eq}}_t}{\mathrm d\phi}
    &=-\gamma e^{r(T-t)}\eta\sigma x_t^*|y_t^*|
    +lx_t^*\frac{\mathrm d\theta_t^*}{\mathrm d\phi},
    \nonumber\\
    \frac{\mathrm d\mathrm{CE}^{H,\mathrm{eq}}_t}{\mathrm d\phi}
    &=-lx_t^*\frac{\mathrm d\theta_t^*}{\mathrm d\phi},
    \nonumber\\
    \frac{\mathrm d\mathrm{CE}^{T,\mathrm{eq}}_t}{\mathrm d\phi}
    &=-\gamma e^{r(T-t)}\eta\sigma x_t^*|y_t^*|.
    \label{eq:incidence} 
\end{align} 
At regime boundaries, the derivatives are understood as the corresponding feasible one-sided derivatives, with right derivatives at $\phi=0$. When underwriting is positive, the insurer's equilibrium local certainty equivalent increases to first order if and only if 
\begin{equation}
    l\frac{\mathrm d\theta_t^*}{\mathrm d\phi} >\gamma e^{r(T-t)}\eta\sigma|y_t^*|. \label{eq:gain-criterion}
\end{equation} 
\end{proposition}

\begin{proof}
See Appendix~\ref{app:proof:incidence}.
\end{proof}

The insurer's equilibrium certainty-equivalent response combines the direct ambiguity cost with the effect of the changing loading. The latter is the marginal price effect on optimized certainty equivalent, rather than the total derivative of premium receipts. When underwriting is positive, a higher loading offsets part of the direct cost, and the insurer benefits overall only if pass-through is sufficiently large. The household faces the opposite price effect, which cancels in total local surplus. Consequently, total local surplus decreases while joint exposure is positive and remains constant after joint exposure reaches zero.

\section{Numerical Illustration}\label{sec:quantitative}

We use numerical examples to compare financial investment at fixed underwriting with the adjustments that arise in competitive equilibrium. The analysis also illustrates the equilibrium regimes and how ambiguity costs are distributed between insurer and household.

\subsection{Numerical Benchmark}

Time is measured in years, and monetary quantities are expressed in arbitrary but consistent units. We use the following illustrative benchmark: 
\begin{gather*}
l=2,\quad \eta=0.28,\quad r=0.015,\quad \mu=0.035,\quad \sigma=0.18,\\
\alpha=\gamma=2,\quad T=5,\quad t=0,\quad \rho_0=0.
\end{gather*}

These parameters describe a representative insurance pool and investment opportunity rather than a calibration to a particular insurer. For comparison across exposures, we report the normalized financial position $\frac{\sigma y}{\eta}$ and its absolute value. These measures are invariant to the monetary unit but should not be interpreted as portfolio wealth shares. 

The rolling correlation estimates in Figure~\ref{fig:dependence} yield a median confidence-interval half-width of approximately $0.23$, which we use as an empirical reference for the magnitude of dependence uncertainty. In the numerical experiments, the correlation interval is centered at $\rho_0=0$ and widened while all other parameters remain fixed. The resulting interval is a model input rather than an estimated ambiguity set or a confidence interval for a particular date. 

\subsection{Conditional Financial Exposure Under Ambiguity}

Figure~\ref{fig:conditional} shows optimal financial investment for three fixed underwriting shares, $x\in\{0.25,0.5,0.75\}$. With a zero correlation midpoint and a positive expected excess return, the insurer initially holds a long financial position. As the ambiguity interval widens, the adverse upper correlation bound increases, reducing the attractiveness of combining underwriting and financial investment. The position declines linearly until investment reaches zero. On the positive-investment branch, the slope of the normalized position $\frac{\sigma y}{\eta}$ with respect to $\phi$ is $-x$, so a larger insurance book produces a stronger investment response.

\begin{figure}[!htbp]
\centering
\includegraphics[width=\linewidth]{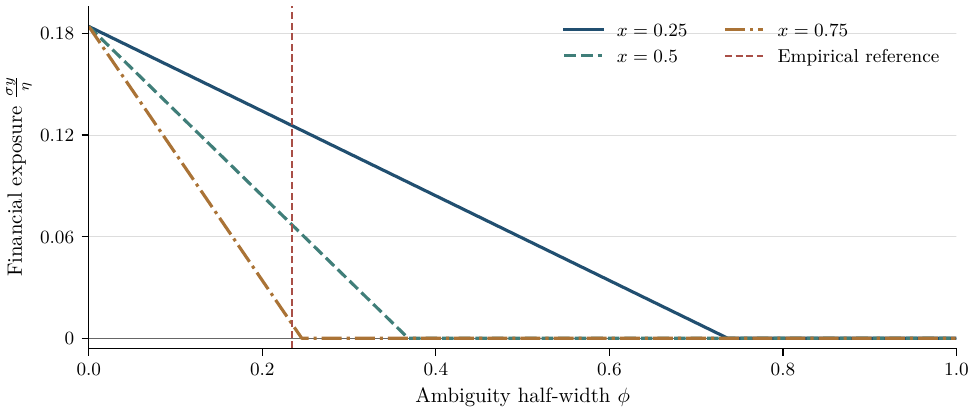}
\caption{Conditional financial exposure at fixed underwriting. The signed normalized financial position $\frac{\sigma y}{\eta}$ is shown for three underwriting shares at $\rho_0=0$. Kinks indicate transitions to zero investment, and the vertical reference line marks the empirical half-width.}
\label{fig:conditional}
\end{figure}

\subsection{Equilibrium Exposure Adjustment}\label{sec:numerical-exposure}

When underwriting adjusts endogenously, the financial investment response need not follow the fixed-book comparison. Figure~\ref{fig:joint-exposure} shows that coverage initially decreases and subsequently recovers, while absolute financial exposure falls toward zero. Their product, joint exposure, declines throughout, consistent with Proposition~\ref{prop:joint-exposure}. At the empirical reference width, the reduction in joint exposure is substantial relative to the changes in coverage and the premium per unit insured, $(1+\theta)l$. The equilibrium loading initially rises and then falls. Once financial investment reaches zero, underwriting and the loading remain unchanged under further widening.

\begin{figure}[!htbp]
\centering
\includegraphics[width=\linewidth]{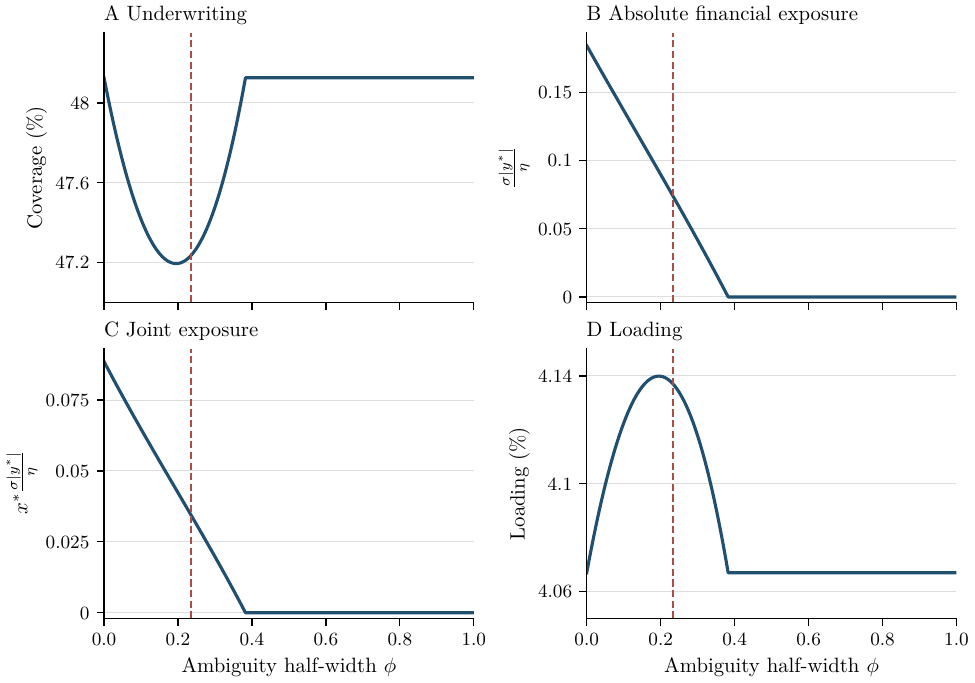}
\caption{Equilibrium adjustment under fixed-midpoint widening at $\rho_0=0$. The panels show coverage in percent, normalized absolute financial exposure $\frac{\sigma|y_t^*|}{\eta}$, normalized joint exposure $x_t^*\frac{\sigma|y_t^*|}{\eta}$, and insurance loading in percent. Vertical lines indicate the empirical reference. The panels have different vertical scales, and the excluded correlation endpoint is not evaluated.}
\label{fig:joint-exposure}
\end{figure} 

\FloatBarrier

The common-correlation evaluation in Corollary~\ref{cor:common-correlation} helps distinguish the effect of reallocating exposures from the additional cost of worst-case evaluation. At the fixed midpoint, the two evaluations satisfy 
\begin{equation*}
    \mathrm{CE}^{T}_t(x_t^*,y_t^*;\rho_0,\rho_0) =\mathrm{CE}^{T,\mathrm{eq}}_t +\gamma e^{r(T-t)}\eta\sigma\phi x_t^*|y_t^*|.
\end{equation*}

Along the benchmark path, total local certainty equivalent also declines when the changing equilibrium allocations are evaluated at the same midpoint correlation. The difference between this common-correlation evaluation and robust equilibrium surplus is the ambiguity penalty at the resulting allocation.

The adjustment of financial investment can be different when underwriting responds more strongly. To illustrate Remark~\ref{remark:investment-rise}, we reduce household risk aversion to $\alpha=0.2$, keeping all other benchmark parameters unchanged. Figure~\ref{fig:investment-rise} shows that financial investment initially declines but subsequently increases as underwriting contracts. Joint exposure nevertheless continues to decrease. When underwriting reaches zero, financial investment becomes the stand-alone optimal position and remains constant under further widening.

\begin{figure}[!htbp]
\centering
\includegraphics[width=\linewidth]{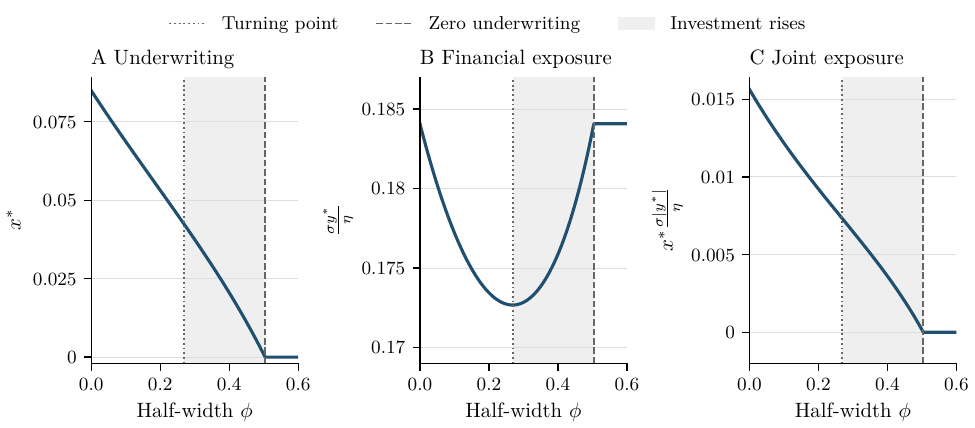}
\caption{Equilibrium investment adjustment with $\alpha=0.2$ and all other benchmark parameters unchanged. The panels display underwriting $x_t^*$, signed normalized financial exposure $\frac{\sigma y_t^*}{\eta}$, and normalized joint exposure $x_t^*\frac{\sigma|y_t^*|}{\eta}$. The investment turning point and zero-underwriting boundary are determined by Remark~\ref{remark:investment-rise} and Theorem~\ref{thm:regimes}, respectively. Vertical scales differ across panels.} 
\label{fig:investment-rise}
\end{figure}

\subsection{Regime Structure and Incidence}

The preceding illustrations hold the correlation midpoint fixed. Figure~\ref{fig:regime-map} extends the comparison by varying both the midpoint and the ambiguity half-width over the admissible correlation domain. Four equilibrium regimes arise under the benchmark parameters: full underwriting, interior underwriting with long investment, pure underwriting, and interior underwriting with a short hedge. Zero underwriting is not attainable in this parameterization because its threshold requires an upper correlation bound of at least $\frac{\alpha\eta\sigma}{\mu-r}=5.04$, which exceeds the admissible range. The empirical reference lies in the interior long-investment region.

\begin{figure}[!htbp]
\centering
\includegraphics[width=\linewidth]{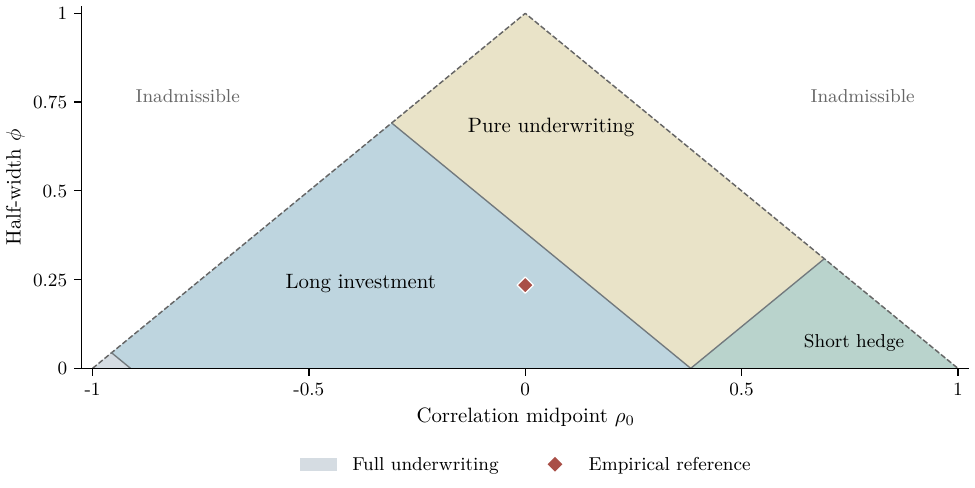}
\caption{Equilibrium regimes across correlation midpoint and ambiguity half-width. The admissible region satisfies $-1<\rho_0-\phi\le\rho_0+\phi<1$, with the two sloping boundaries excluded. Distinct regions correspond to the equilibrium cases in Theorem~\ref{thm:regimes}; the empirical reference is marked at $\rho_0=0$. Areas outside the triangular domain are inadmissible.}
\label{fig:regime-map}
\end{figure}

\FloatBarrier

Figure~\ref{fig:incidence} illustrates how the direct cost of correlation ambiguity and the adjustment of insurance loading affect the insurer and the household. Along the benchmark path, the increase in loading does not fully offset the insurer's direct ambiguity cost before financial investment ceases. The household's marginal effect changes sign when the loading begins to decline, while total local surplus decreases whenever joint exposure is positive. Once financial investment reaches zero, further widening leaves the allocation, loading, and local certainty equivalents unchanged.

\begin{figure}[!htbp]
\centering
\includegraphics[width=\linewidth]{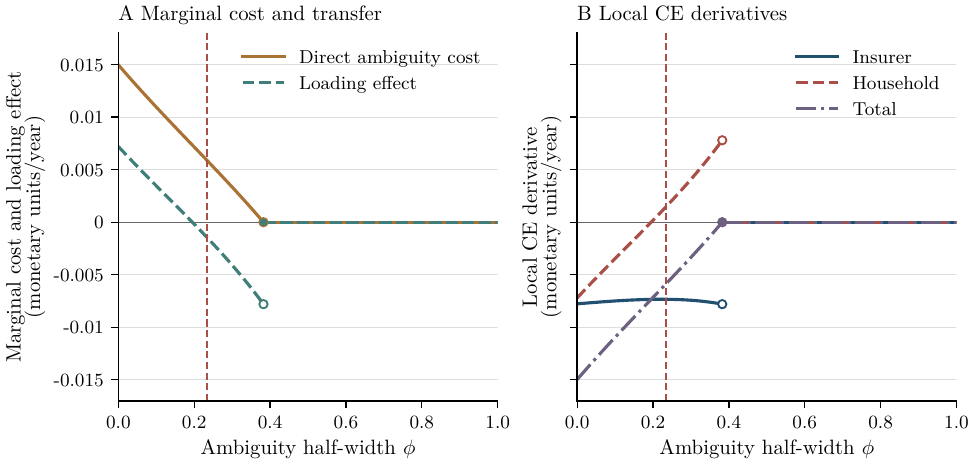}
\caption{Marginal ambiguity cost and incidence at $\rho_0=0$. Panel A compares the direct ambiguity cost with the loading effect, while Panel B reports the derivatives of insurer, household, and total local certainty equivalents. Both panels use monetary units per year. The empirical reference is indicated; open circles at specialization show left limits, and filled points at zero indicate right derivatives. Curves are not joined across the derivative jump.}
\label{fig:incidence}
\end{figure}

\FloatBarrier

\section{Conclusion}\label{sec:conclusion}

Insurers need not reduce financial investment to limit their exposure to uncertain dependence. When the correlation interval widens around a fixed midpoint, underwriting may contract enough for investment to rise, even as joint exposure strictly declines while positive. Once joint exposure vanishes, further feasible widening leaves the equilibrium allocation and loading unchanged. Competitive insurance pricing redistributes ambiguity costs between insurer and policyholder, but does not eliminate the decline in their combined local certainty-equivalent surplus while joint exposure remains positive. Evaluating allocations at a common reference correlation also separates the effect of risk reallocation from that of worst-case evaluation. Future research could examine how capital and liquidity constraints affect these adjustments. Allowing insurers to learn about risk dependence over time could also shed light on the dynamics of insurance supply and pricing.

\clearpage
\phantomsection\label{sec:references}
\bibliographystyle{apalike}
\bibliography{references}

\clearpage
\appendix

\section{Proofs}\label{app:proofs}

\subsection{Proof of Theorem~\ref{thm:verification}}\label{app:proof:verification}

\begin{proof}\label{proof:thm:verification}
For every $(v,w)\in\mathbb R^2$ and $\rho\in[\underline{\rho},\overline{\rho}]$, 
\begin{equation*}
    \eta^2v^2+2\rho\eta\sigma vw+\sigma^2w^2 \ge \left(1-\max\{|\underline{\rho}|,|\overline{\rho}|\}\right) (\eta^2v^2+\sigma^2w^2).
\end{equation*} 
Since the correlation bounds lie strictly between $-1$ and $1$, the functions $F_t^\rho$ are strongly concave with a common positive modulus. Their pointwise minimum $F_t$ inherits this property by applying the strong-concavity inequality before taking the minimum over $\rho$. The same quadratic bound, together with boundedness of $\theta$, gives coercivity in $y$, uniformly over $x\in[0,1]$ and $t\in[0,T]$. Hence $F_t$ has a unique maximizer. The conditional investment formula \eqref{eq:soft-threshold} yields the uniform bound \eqref{eq:optimizer-bound}. Continuity of the objective and uniqueness then imply continuity of the maximizer and optimized local certainty equivalent by taking convergent subsequences of optimizers.

For an admissible policy $(x,y)\in\mathcal A_t$ and correlation process $\rho\in\mathcal N_t$, suppress the fixed initial data and write the discounted martingale part of insurer wealth as 
\begin{equation*}
    \mathcal M_s =\int_t^s e^{-r(u-t)} \left(\eta x_u\,\mathrm dB_u^I+\sigma y_u\,\mathrm dW_u^{S,\rho}\right).
\end{equation*}
The wealth equation has the unique strong solution 
\begin{equation*}
    M_s=e^{r(s-t)} \left[ m+\int_t^s e^{-r(u-t)} \left(\theta_u l x_u+(\mu-r)y_u\right)\mathrm du +\mathcal M_s \right].
\end{equation*}

By the bounded underwriting share and the uniform correlation bounds, there is a finite constant $C>0$, independent of $\rho$, such that 
\begin{equation}
    \begin{aligned}
    \sup_{t\le s\le T}|M_s| &\le C\left( 1+|m|+ \left(\int_t^T y_u^2\,\mathrm du\right)^{1/2} +\sup_{t\le s\le T}|\mathcal M_s|
    \right),\\  \langle\mathcal M\rangle_T
    &\le C\left(1+\int_t^T y_u^2\,\mathrm du\right). 
    \end{aligned}
\label{eq:energy-bounds}
\end{equation}

Applying the exponential martingale inequality to $\mathcal M$ and $-\mathcal M$, and splitting according to whether $\langle\mathcal M\rangle_T\le z/(4a)$, gives, for any $a,z>0$,

\begin{equation*}
    \mathbb P\left(\sup_{t\le s\le T}|\mathcal M_s|\ge z\right) \le \left(2+\mathbb E e^{8a^2\langle\mathcal M\rangle_T}\right)e^{-2az}.
\end{equation*}

By \eqref{eq:admissibility} and \eqref{eq:energy-bounds}, the expectation on the right is finite uniformly over $\rho\in\mathcal N_t$. Since $a$ can be chosen arbitrarily large, integration of the tail bound gives exponential moments of $\sup_s|\mathcal M_s|$ of every positive order. The first inequality in \eqref{eq:energy-bounds}, together with Young's and H\"older's inequalities, therefore yields 
\begin{equation}
    \sup_{\rho\in\mathcal N_t} \mathbb E\exp\left(b\sup_{t\le s\le T}|M_s|\right)<\infty, \qquad b>0,  \label{eq:wealth-moments}
\end{equation} 
for each admissible insurer policy. The admissible classes are also stable under finite stopping-time and event pasting. For insurer policies, the integrated squared investment of a pasted policy is bounded by the sum of the original energies, so \eqref{eq:admissibility} is preserved by H\"older's inequality.

The optimized local certainty equivalent is continuous and bounded. Consequently, the candidate value function in \eqref{eq:value} belongs to $C^{1,2}$, with one-sided time derivatives at the endpoints, and satisfies $V_m=-\gamma e^{r(T-t)}V>0$. Writing $\pi=(x,y)$, direct differentiation and the wealth dynamics give 
\begin{equation}
    V_t+\mathcal L^{\pi,\rho}V =V_m\left[F_t^\rho(x,y;\theta_t)-\mathrm{CE}^{I}_t(\theta_t)\right], \qquad V(T,m)=U(m).
\label{eq:generator}
\end{equation}

For any admissible policy $\pi$, define the progressively measurable worst-case selector 
\begin{equation}
    \rho_s^\pi=
        \begin{cases}
        \overline{\rho}, & x_sy_s>0,\\
        \underline{\rho}, & x_sy_s\le0. 
        \end{cases}
    \label{eq:adverse-selector}
\end{equation} 
This selector attains the minimum in the local objective, with either endpoint admissible when $x_sy_s=0$. Hence the drift in \eqref{eq:generator} is nonpositive under $(\pi,\rho^\pi)$. Conversely, for the deterministic pointwise maximizer $\pi^*$, the drift is nonnegative under every admissible correlation process and vanishes under $\rho^{\pi^*}$.

These drift comparisons can be integrated without localization errors. Indeed, \eqref{eq:wealth-moments}, \eqref{eq:admissibility}, and H\"older's inequality imply 
\begin{equation*}
    \begin{aligned}
        &\mathbb E\int_t^T V_m(s,M_s)^2 \left[ \eta^2x_s^2+\sigma^2y_s^2+2\rho_s\eta\sigma x_sy_s \right]\mathrm ds \le C\mathbb E\left[ e^{C\sup_{t\le s\le T}|M_s|} \left(1+\int_t^T y_s^2\,\mathrm ds\right) \right]<\infty.
    \end{aligned}
\end{equation*} 
Thus the stochastic integrals in It\^o's formula are square-integrable martingales. The drift is absolutely integrable by the same estimate, since $|F_s^\rho(x_s,y_s;\theta_s)-\mathrm{CE}^{I}_s(\theta_s)|\le C(1+y_s^2)$, and terminal utility is integrable.

Conditional It\^o's formula therefore gives, for every admissible $\pi$, 
\begin{equation*}
    \operatorname*{ess\,inf}_{\rho\in\mathcal N_t} \mathbb E\left[U(M_T^{\pi,\rho})\mid\mathcal F_t\right] \le \mathbb E\left[U(M_T^{\pi,\rho^\pi})\mid\mathcal F_t\right] \le V(t,m), 
\end{equation*} 
where the fixed restart date and initial wealth are suppressed in the wealth notation. Taking the essential supremum over $\pi$ gives the upper bound. The pointwise maximizing policy $\pi^*$ is bounded and deterministic by \eqref{eq:optimizer-bound}, hence admissible. Moreover, for every $\rho\in\mathcal N_t$, 
\begin{equation*}
    \mathbb E\left[U(M_T^{\pi^*,\rho})\mid\mathcal F_t\right] \ge V(t,m) = \mathbb E\left[U(M_T^{\pi^*,\rho^{\pi^*}})\mid\mathcal F_t\right].
\end{equation*} 
The upper and lower bounds establish the robust value and attainment without interchanging the two optimizations.

Finally, consider continuous deterministic policies. If such a policy differs from the pointwise maximizer at some date before $T$, strict concavity and continuity give a strictly positive local objective gap on an interval of positive length. Under its worst-case selector, the drift in \eqref{eq:generator} is therefore strictly negative on that interval and nonpositive elsewhere. The preceding integrability argument implies a strictly lower expected terminal utility. Continuity rules out a discrepancy at $T$ alone. Hence, among continuous deterministic policies, dynamic optimality is equivalent to pointwise maximization of $F_t$. 
\end{proof}

\subsection{Proof of Theorem~\ref{thm:regimes}}\label{app:proof:regimes}

\begin{proof} \label{proof:thm:regimes}
We first establish a sufficient condition for robust optimality. If a feasible allocation $(x_t^*,y_t^*)$ maximizes $F_t^\rho$ for some $\rho\in[\underline{\rho},\overline{\rho}]$ and satisfies $F_t(x_t^*,y_t^*;\theta)=F_t^\rho(x_t^*,y_t^*;\theta)$, then, for every feasible $(x,y)$, 
\begin{equation}
    F_t(x,y;\theta)\le F_t^\rho(x,y;\theta)\le F_t^\rho(x_t^*,y_t^*;\theta)=F_t(x_t^*,y_t^*;\theta). \label{eq:global-certificate}
\end{equation} 
We use this observation together with the first-order conditions and market clearing to establish the five equilibrium cases.

For full underwriting, $x_t^*=1$ requires $\theta_t^*=0$ by market clearing. Investment cannot be zero, since reducing underwriting would increase $F_t(x,0;0)=-\frac{\gamma e^{r(T-t)}\eta^2x^2}{2}$ at $x=1$. Nor can investment be negative: the conditional investment rule would require $\underline{\rho}>0$, in which case the underwriting derivative at the financial optimum is 
\begin{equation*}
    -\gamma e^{r(T-t)}\eta^2\left[1-\underline{\rho}^{\,2}+\frac{\underline{\rho}(\mu-r)e^{-r(T-t)}}{\gamma\eta\sigma}\right]<0.
\end{equation*} 
Thus full underwriting requires positive investment, with $\overline{\rho}$ as the adverse correlation. The financial first-order condition gives \eqref{eq:regime-full}, while the underwriting upper-bound condition is 
\begin{equation}
    -\gamma e^{r(T-t)}\eta^2\left[1-\overline{\rho}^{\,2}+\frac{\overline{\rho}(\mu-r)e^{-r(T-t)}}{\gamma\eta\sigma}\right]\ge0. \label{eq:capacity-kkt}
\end{equation} 
The quadratic expression in brackets has one root in $(-1,0)$, given by \eqref{eq:regime-full-condition}, and another greater than one. Hence \eqref{eq:capacity-kkt} is equivalent to \eqref{eq:regime-full-condition} on the admissible correlation domain. This condition also ensures positive investment. The constrained first-order conditions maximize the strictly concave $F_t^{\overline{\rho}}$, and \eqref{eq:global-certificate} establishes robust optimality.

For interior underwriting with positive investment, the adverse correlation is $\overline{\rho}$. The two first-order conditions, together with market clearing, yield 
\begin{equation}
    \begin{aligned}
        \eta^2x_t^*+\overline{\rho}\eta\sigma y_t^*
        &=\frac{\theta_t^*l}{\gamma}e^{-r(T-t)},\\
        \overline{\rho}\eta\sigma x_t^*+\sigma^2y_t^*
        &=\frac{\mu-r}{\gamma}e^{-r(T-t)},\\
        x_t^*&=1-\frac{l\theta_t^*}{\alpha\eta^2}.
    \end{aligned}
\label{eq:equilibrium-long-system}
\end{equation} 
Solving this system gives \eqref{eq:regime-long}. Its denominator is positive because $|\overline{\rho}|<1$. The conditions $x_t^*<1$, $x_t^*>0$, and $y_t^*>0$ are equivalent, respectively, to the upper correlation bound exceeding the full-underwriting threshold, remaining below $\frac{\alpha\eta\sigma}{\mu-r}$, and remaining below the pure-underwriting threshold. These are precisely the conditions in \eqref{eq:regime-long-condition}. The resulting interior allocation satisfies the first-order conditions of the strictly concave $F_t^{\overline{\rho}}$, so \eqref{eq:global-certificate} establishes sufficiency.

For pure underwriting, $y_t^*=0$, the underwriting first-order condition and market clearing give \eqref{eq:regime-pure}, with $0<x_t^*<1$. The financial subgradient condition \eqref{eq:financial-subgradient} reduces to \eqref{eq:regime-pure-condition}. Conversely, whenever this condition holds, the correlation 
\begin{equation}
    \rho=\frac{(\mu-r)\left(\frac{1}{\alpha}+\frac{1}{\gamma}e^{-r(T-t)}\right)}{\eta\sigma} \label{eq:pure-support}
\end{equation} 
belongs to the ambiguity interval. At the proposed allocation, both first-order conditions of $F_t^\rho$ hold, and its value agrees with $F_t$ because investment is zero. Strict concavity and \eqref{eq:global-certificate} establish robust optimality.

For interior underwriting with negative investment, the adverse correlation is $\underline{\rho}$. Replacing $\overline{\rho}$ by $\underline{\rho}$ in the first two equations of \eqref{eq:equilibrium-long-system} gives \eqref{eq:regime-short}. The financial position is negative precisely under \eqref{eq:regime-short-condition}, which implies $\underline{\rho}>0$. Moreover, 
\begin{equation}
    \frac{(\mu-r)\left(\frac{1}{\alpha}+\frac{1}{\gamma}e^{-r(T-t)}\right)}{\eta\sigma}\frac{\alpha\eta\sigma}{\mu-r} =1+\frac{\alpha}{\gamma}e^{-r(T-t)}>1. \label{eq:ordering}
\end{equation} 
Since the pure-underwriting threshold lies below $\underline{\rho}<1$, this identity places the zero-underwriting threshold above one. The underwriting share in \eqref{eq:regime-short} is therefore positive and, since $\underline{\rho}>0$, strictly below one. The loading is also interior. The first-order conditions of $F_t^{\underline{\rho}}$ and \eqref{eq:global-certificate} again establish sufficiency.

For zero underwriting, $x_t^*=0$, market clearing requires $\theta_t^*=\frac{\alpha\eta^2}{l}$. Financial optimization gives the positive position in \eqref{eq:regime-zero}. The right underwriting derivative is $\alpha\eta^2-\frac{\overline{\rho}\eta(\mu-r)}{\sigma}$, which is nonpositive exactly under \eqref{eq:regime-zero-condition}. The financial first-order condition and the underwriting boundary condition maximize $F_t^{\overline{\rho}}$. Since the objective is independent of correlation at zero underwriting, \eqref{eq:global-certificate} establishes robust optimality.

It remains to verify that the cases are exhaustive and mutually exclusive. The full-underwriting threshold lies in $(-1,0)$, whereas the pure- and zero-underwriting thresholds are positive. By \eqref{eq:ordering}, the latter two cannot both lie in $(-1,1)$. The upper correlation endpoint first determines whether full or zero underwriting applies. Outside these boundary cases, an upper endpoint below the pure-underwriting threshold gives positive investment; otherwise, the interval either contains that threshold, giving pure underwriting, or lies entirely above it, giving negative investment. The strict inequalities and boundary assignments in Theorem~\ref{thm:regimes} exclude overlap. Each case satisfies market clearing and robust optimality, and the unique clearing loading established in Section~\ref{sec:supply-clearing} gives equilibrium uniqueness.

The formulas also agree at every reachable regime boundary. At the full-underwriting interface, equality in \eqref{eq:capacity-kkt} gives $x_t^*=1$ in the positive-investment formula, with zero loading and the financial position in \eqref{eq:regime-full}. At either pure-underwriting interface, the active correlation equals \eqref{eq:pure-support}, so the adjacent interior solution has $y_t^*=0$ and reduces to \eqref{eq:regime-pure}. At the zero-underwriting interface, $\overline{\rho}(\mu-r)=\alpha\eta\sigma$ makes the positive-investment underwriting share vanish, with the loading and investment given by \eqref{eq:regime-zero}.

Finally, continuity of insurer supply in date and loading, together with uniqueness of market clearing, implies continuity of the equilibrium loading. Continuity of the pointwise optimizer then gives continuous underwriting and investment paths and optimized local insurer certainty equivalents. The investment bound \eqref{eq:optimizer-bound} makes the deterministic equilibrium policy admissible, while Theorem~\ref{thm:verification} establishes insurer optimality against the full progressive admissible class. Household optimality follows from \eqref{eq:demand} and market clearing. 
\end{proof}

\subsection{Proof of Theorem~\ref{thm:planner}}\label{app:proof:planner}

\begin{proof}\label{proof:thm:planner}
At a common coverage level $x$, the insurance loading enters the insurer's and household's local certainty equivalents with opposite signs. Their sum gives \eqref{eq:local-total-ce}. The uniform strong concavity and coercivity established in Appendix~\ref{app:proof:verification}, together with the concave household risk term, imply that this surplus has a unique maximizer on $[0,1]\times\mathbb R$. Continuity of the maximizer in the model parameters follows from the same compactness and uniqueness argument. 

At a competitive equilibrium, the insurer maximizes $F_t(x,y;\theta_t^*)$, while the household maximizes $\mathrm{CE}^H_t(x,\theta_t^*)$. Since their choices agree at market clearing, adding the two optimality inequalities shows that $(x_t^*,y_t^*)$ maximizes the local surplus.

Conversely, let $(x^\dagger,y^\dagger)$ maximize \eqref{eq:local-total-ce} and set $\theta^\dagger=\frac{\alpha\eta^2}{l}(1-x^\dagger)$. This loading belongs to the maintained interval and satisfies 
\begin{equation}
    -\theta^\dagger l+\alpha\eta^2(1-x^\dagger)=0. \label{eq:household-support}
\end{equation} 
Strict concavity of the household objective makes $x^\dagger$ its unique optimal coverage, including at the coverage boundaries. Moreover, at this loading, 
\begin{equation*}
    \mathrm{CE}^H_t(x,\theta^\dagger)-\mathrm{CE}^H_t(x^\dagger,\theta^\dagger) =-\frac{\alpha\eta^2}{2}(x-x^\dagger)^2.
\end{equation*}

For any feasible $(x,y)$ and $0<\varepsilon<1$, apply surplus optimality to the allocation $(1-\varepsilon)(x^\dagger,y^\dagger)+\varepsilon(x,y)$. Using the preceding identity and concavity of $F_t$ gives 
\begin{equation*}
    \varepsilon\left[F_t(x,y;\theta^\dagger)-F_t(x^\dagger,y^\dagger;\theta^\dagger)\right] \le\frac{\alpha\eta^2}{2}\varepsilon^2(x-x^\dagger)^2.
\end{equation*} 
Dividing by $\varepsilon$ and letting $\varepsilon\downarrow0$ shows that $(x^\dagger,y^\dagger)$ maximizes the insurer's objective at $\theta^\dagger$. Thus both agents optimize at the same loading, and insurance supply equals household demand. Continuity of the maximizing allocation ensures a continuous deterministic supporting loading, while Theorem~\ref{thm:verification} establishes the insurer's dynamic optimality. The equilibrium is therefore decentralized by \eqref{eq:decentralizing-loading}.

Finally, enlarging the correlation interval weakly lowers the worst-case local surplus at every fixed allocation, since Nature minimizes over a larger set while the household risk term remains unchanged. Maximizing over the common feasible set preserves this ordering, proving the stated weak monotonicity of equilibrium local surplus. 
\end{proof}

\subsection{Proof of Proposition~\ref{prop:fixedprice}}\label{app:proof:fixedprice}

\begin{proof}\label{proof:prop:fixedprice} 
Enlarging the correlation interval weakly decreases the worst-case objective at every fixed allocation. Maximizing over the same feasible set therefore weakly decreases the insurer's optimized local certainty equivalent.

To establish the envelope derivatives, rewrite the local objective as 
\begin{equation}
    \begin{aligned}
        F_t(x,y;\theta,\underline{\rho},\overline{\rho}) ={}&\theta lx+(\mu-r)y -\frac{\gamma e^{r(T-t)}}{2}(\eta^2x^2+\sigma^2y^2)\\
        &-\gamma e^{r(T-t)}\eta\sigma x \left[\overline{\rho}\max(y,0)+\underline{\rho}\min(y,0)\right].
    \end{aligned}
    \label{eq:affine-endpoints}
\end{equation} 
For any fixed allocation, this expression is affine in the loading and correlation endpoints, with parameter derivatives continuous in $(x,y)$, including at $y=0$. The optimizer is unique by Theorem~\ref{thm:verification}, and the uniform investment bound \eqref{eq:optimizer-bound} places optimizers in a common compact set for nearby admissible parameters. Continuity of the objective and uniqueness then imply continuity of the optimizer.

The envelope formula follows by comparing the perturbed objective at the original optimizer with the original objective at the perturbed optimizer. These evaluations bound the change in optimized certainty equivalent. By affinity in the parameters and continuity of the optimizer, both bounds have the same first-order term, given by the parameter derivative of \eqref{eq:affine-endpoints} at the optimum. Taking these derivatives yields \eqref{eq:envelopes}. At the boundaries of the admissible parameter domain, the same argument applies to feasible one-sided perturbations.
\end{proof}

\subsection{Proof of Proposition~\ref{prop:incidence}}\label{app:proof:incidence}

\begin{proof}\label{proof:prop:incidence}
Along the fixed-midpoint widening in \eqref{eq:symmetric-experiment}, the correlation endpoints satisfy $\frac{\mathrm d\overline{\rho}}{\mathrm d\phi}=1$ and $\frac{\mathrm d\underline{\rho}}{\mathrm d\phi}=-1$. Applying the envelope formulas \eqref{eq:envelopes} at the equilibrium allocation and using $\max(y,0)-\min(y,0)=|y|$ yields the insurer's derivative in \eqref{eq:incidence}. For the household, \eqref{eq:ce-levels} and market clearing give $\frac{\mathrm d\mathrm{CE}^{H,*}_t}{\mathrm d\theta}=-lx_t^*$ at $\theta=\theta_t^*$. Adding the two derivatives gives the change in total local certainty equivalent. When $x_t^*>0$, the condition for the insurer's derivative to be positive reduces to \eqref{eq:gain-criterion}.

These identities hold within each smooth equilibrium regime. At regime boundaries, the continuity of the equilibrium allocation established in Theorem~\ref{thm:regimes}, together with the feasible one-sided envelope derivatives proved in Appendix~\ref{app:proof:fixedprice}, gives the corresponding one-sided formulas. At $\phi=0$, the derivatives are taken from the right. This covers the zero-investment kink and the underwriting boundaries.
\end{proof}

\end{document}